\documentclass[runningheads]{llncs}
\usepackage{amssymb,amsmath,graphicx}
\usepackage{xspace}
\usepackage[bookmarks=false,pdfpagelabels=false]{hyperref}
\usepackage{xcolor}
\hypersetup{breaklinks,pdfdisplaydoctitle,colorlinks,linkcolor=red!50!black,citecolor=green!50!black}
\usepackage{tikz}

\usepackage{todonotes}

\tikzset{label distance=-1pt,
	vertex/.style={inner sep=.2em,circle,fill=black,draw},
	edge/.style={thick}}

\newcommand{\lb}{1.4521}
\newcommand{\sep}{\mathsf{sep}}
\newcommand{\pmc}{\mathsf{pmc}}

\begin{document}

\title{On the Number of Minimal Separators in Graphs}

\author{Serge Gaspers\inst{1,2}\and Simon Mackenzie\inst{1,2}
\institute{The University of New South Wales, Australia, \email{\{sergeg,simonwm\}@cse.unsw.edu.au}
 \and NICTA, Australia}
}

\maketitle

\begin{abstract}
We consider the largest number of minimal separators a graph on $n$ vertices can have at most.
\begin{itemize}
	\item We give a new proof that this number is in $O\left( \left( \frac{1+\sqrt{5}}{2}\right)^n\cdot n \right)$.
	\item We prove that this number is in $\omega\left( \lb^n \right)$, improving on the previous best lower bound of $\Omega(3^{n/3}) \subseteq \omega(1.4422^n)$.
\end{itemize}
This gives also an improved lower bound on the number of potential maximal cliques in a graph.
We would like to emphasize that our proofs are short, simple, and elementary. 
\end{abstract}

\section{Introduction}

For a graph $G=(V,E)$, and two vertices $a,b\in V$, a vertex subset $S\subseteq V\setminus \{a,b\}$ is an \emph{$(a,b)$-separator} if $a$ and $b$ are in different connected components of $G-S$, the graph obtained from $G$ by removing the vertices in $S$.
An $(a,b)$-separator is \emph{minimal} if it does not contain another $(a,b)$-separator as a subset.
A vertex subset $S\subset V$ is a \emph{minimal separator} in $G$ if it is a minimal $(a,b)$-separator for some pair of distinct vertices $a,b\in V$.

By $\sep(G)$, we denote the number of minimal separators in the graph $G$.
By $\sep(n)$, we denote the maximum number of minimal separators, taken over all graphs on $n$ vertices.

Potential maximal cliques are closely related to minimal separators, especially in the context of chordal graphs.
A graph is \emph{chordal} if every induced cycle has length 3.
A \emph{triangulation} of a graph $G$ is a chordal supergraph of $G$ obtained by adding edges.
A graph $H$ is a \emph{minimal triangulation} of $G$ if it is a triangulation of $G$ and $G$ has no other triangulation that is a subgraph of $H$.
A vertex set is a \emph{potential maximal clique} in $G$ if it is a maximal clique in at least one minimal triangulation of $G$.

By $\pmc(G)$, we denote the number of potential maximal cliques in the graph $G$.
By $\pmc(n)$, we denote the maximum number of potential maximal cliques, taken over all graphs on $n$ vertices.

Minimal separators and potential maximal cliques have been studied extensively~\cite{BerryBC00,BouchitteT01,BouchitteT02,FominKTV08,FominV12,Heggernes06,KloksK98,ParraS97,ShenL97,Villanger06}.
Upper bounds on $\sep(n)$ are used to upper bound the running time of algorithms for enumerating all minimal separators~\cite{BerryBC00,KloksK98,ShenL97}.
Bounds on both $\sep(n)$ and $\pmc(n)$ are used in analyses of algorithmic running times
for computing the treewidth and minimum fill-in of a graph~\cite{BouchitteT02,FominKTV08,FominV12}, and for
computing a maximum induced subgraph isomorphic to a graph from a family of bounded treewidth graphs~\cite{FominV10}.

\paragraph{Our results.}
Fomin et al.~\cite{FominKTV08} proved that $\sep(n) \in O\left( 1.7087^n \right)$.
Fomin and Villanger~\cite{FominV12} improved the upper bound and showed that $\sep(n) \in O\left( \rho^n \cdot n\right)$, where $\rho = \frac{1+\sqrt{5}}{2} = 1.6180\ldots$ \footnote{The bound stated in \cite{FominV12} is $O(1.6181^n)$, but this stronger bound can be derived from their proof.}.
We prove the same upper bound with simpler arguments.

As for lower bounds, it is known~\cite{FominKTV08} that $\sep(n) \in \Omega(3^{n/3})$; see \autoref{fig:melon}.
We improve on this lower bound by giving an infinite family of graphs with $\omega\left( \lb^n \right)$ minimal separators.
This answers an open question raised numerous times~(see, e.g., \cite{FominK10,FominKTV08}),
for example by Fomin and Kratsch \cite[page 100]{FominK10}, who state 
\begin{quotation}
 It is an open question, whether the number of minimal separators
 in every $n$-vertex graph is $O^*(3^{n/3})$.
\end{quotation}
Here, the $O^*$-notation is similar to the $O$-notation, but hides polynomial factors.

As a corollary, we have that there is an infinite family of graphs, all with $\omega(\lb^n)$ potential maximal cliques.
This answers another open question on lower bounds for potential maximal cliques.
For example, Fomin and Villanger \cite{FominV10} state
\begin{quotation}
There are graphs with roughly $3^{n/3}\approx 1.442^n$ potential maximal cliques \cite{FominKTV08}.
Let us remind that by the classical result of Moon and Moser \cite{MoonM65} (see also Miller and Muller \cite{MillerM60}) that the number of maximal cliques in a graph on $n$ vertices is at most $3^{n/3}$.
Can it be that the right upper bound on the number of potential maximal cliques
is also roughly $3^{n/3}$?
By Theorem 3.2, this would yield a dramatic improvement for many moderate exponential
algorithms.
\end{quotation}

\begin{figure}[tb]
	\centering
	\begin{tikzpicture}[xscale=1.1,yscale=0.8]
	\node[vertex,label=above left:$a$] (a) at (0,3) {};
	\node[vertex,label=above right:$b$] (b) at (4,3) {};
	
	\foreach \j in {1,3,4,5} {
		\foreach \i in {1,2,3} {
			\node[vertex] (v-\i-\j) at (\i,\j) {};
			\ifnum\i>1
				\pgfmathparse{int(round(\i-1))}
				\draw (v-\i-\j)--(v-\pgfmathresult-\j);
			\fi
		}
		\draw (a)--(v-1-\j);
		\draw (v-3-\j)--(b);
	}
	\node at (2,2) {$\vdots$};
	\end{tikzpicture}
	\caption{\label{fig:melon}Melon graphs have $\Omega(3^{n/3})$ minimal separators.}
\end{figure}
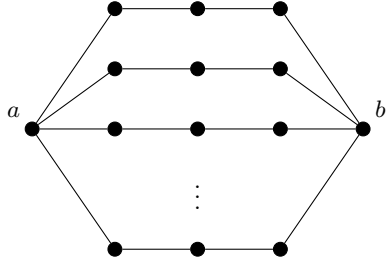

\paragraph{Preliminaries.}
We use standard graph notation from \cite{Diestel10}.
For an edge $uv$ in a graph $G$, we denote 
by $G/uv$ the graph obtained from $G$ by contracting the edge $uv$, i.e., making $u$ adjacent to $N_G(\{u,v\})$ and removing $v$.

\section{Upper bound on the number of minimal separators}

Measure and Conquer is a technique developed for the analysis of exponential time algorithms~\cite{FominGK09}.
Its main idea is to assign a cleverly chosen (sometimes, by solving mathematical programs~\cite{Eppstein06,Gaspers10,GaspersS12}) potential function to the instance -- a so-called \emph{measure} -- to track several features of an instance in solving it.
While developed in the realm of exponential-time algorithms, it has also been used to upper bound the number of extremal vertex sets in graphs (see, e.g., \cite{FominGPR08,FominGPS08}).

Our new proof upper bounding $\sep(n)$ uses a measure that takes into account the number of vertices of the graph and the difference in size between the separated components of the graph.
This simple trick allows us to avoid several complications from \cite{FominV12}, including the use of an auxiliary lemma (Lemma 3.1 in \cite{FominV12}),
fixing the size of the separators, the discussion of ``full components'',
and distinguishing between separators of size at most $n/3$ and at least $n/3$.

\begin{theorem}\label{thm:ub}
	$\sep(n) = O(\rho^n\cdot n)$, where $\rho = \frac{1+\sqrt{5}}{2} = 1.6180...$ is the golden ratio.
\end{theorem}
\begin{proof}
	Let $G=(V,E)$ be any graph on $n$ vertices with $a\in V$.
	For $d\le |V|$, an $[a,d]$-separation is 
	a partition $(A,S,B)$ of $V$ such that 
	\begin{itemize}
		\item $a\in A$,
		\item $G[A]$ is connected,
		\item $S$ is a minimal $(a,b)$-separator for some $b\in B$, and
		\item $|A|\le |B|-d$.
	\end{itemize}
	Let $\sep_{a}(G,d)$ denote the number of $[a,d]$-separations in $G$.
	By symmetry, \linebreak[4]$\sep_{a}(G,0)$ upper bounds
	the number of minimal separators in $G$ up to a factor $O(|V|)$.
	To upper bound $\sep_{a}(G,d)$, we will use the measure
	\begin{align*}
	\mu(G,d) = |V| - d.
	\end{align*}
	The theorem will follow from the claim that $\sep_{a}(G,d) \le \rho^{\mu(G,d)}$ for $0\le d\le |V|$.
	
	If $\mu(G,d)=0$, then $d = |V|$ and $\sep_{a}(G,d)=0$ since there is no $A\subseteq V$ with $|A|\le 0$ and $a\in A$.
	If $d_G(a)=0$, then there is at most one $[a,d]$-separation, which is $(\{a\},\emptyset,V\setminus \{a\})$.
	Therefore, assume $\mu(G,d)\ge 1$, $a$ has at least one neighbor, and assume the claim holds for smaller measures.
	Consider a vertex $u\in N(a)$.
	For every $[a,d]$-separation $(A,S,B)$, either $u\in S$ or $u\in A$.
	Therefore, we can upper bound the $[a,d]$-separations $(A,S,B)$ counted in $\sep_{a,b}(G,d)$ with $u\in S$ by $\rho^{\mu(G-\{u\},d)} = \rho^{\mu(G,d)-1}$, and those with $u \notin S$ by $\rho^{\mu(G/au,d+1)} = \rho^{\mu(G,d)-2}$.
	It remains to observe that $\rho^{\mu(G,d)-1}+\rho^{\mu(G,d)-2} = \rho^{\mu(G,d)}$.
	\qed
\end{proof}

\section{Lower bound on the maximum number of minimal separators}

In the melon graph in \autoref{fig:melon},  each horizontal layer implies a choice between 3 vertices. Each of those choices also `costs' 3 vertices. The new construction improves the bound by adding a vertical choice on top of the horizontal choice. This is achieved by `sacrificing' one of the horizontal choices. This allows us to chose which layer to sacrifice, at the cost of 6 vertices. If it has more than 
$3\cdot 3=9$ layers, then this will give a larger range of choices than if we hadn't eliminated that layer.
\begin{theorem}\label{thm:lb}
	$\sep(n) \in \omega(\lb^n)$.
\end{theorem}
\begin{proof}
	We prove the theorem by exhibiting a family of graphs $\{G_1,G_2,\dots\}$ and lower bounding
	their number of minimal separators.
	
	Let $I = \{1,\dots,6\}$ and $J = \{1,\dots,24\}$.
	The graph $G_1$ is constructed as follows (see \autoref{fig:lb}).
	It has vertex set $V = \{a,b\} \cup \{v_{i,j} : i\in I, j\in J\}$. We denote by $V_i$ the vertex set $\{v_{i,j} : j\in J\}$.
	The edge set $E$ of $G_1$ is
	obtained by first adding the paths $(a,v_{1,j},v_{2,j},v_{3,j})$ and $(v_{4,j},v_{5,j},v_{6,j},b)$ for all $j\in J$, and then adding the edges $\{v_{3,j}v_{4,k} : j,k\in J \text{ and } j\ne k\}$.
	The graph $G_\ell$, $\ell\ge 2$, is obtained from $\ell$ disjoint copies of $G_1$, merging the copies of $a$, and merging the copies of $b$.
	
	Let us now lower bound the minimal $(a,b)$-separators $\mathcal{S}_j$ in $G_1$
	that do not contain any vertex from $\{v_{1,j},v_{2,j},v_{3,j},v_{4,j},v_{5,j},v_{6,j}\}$ for some $j\in J$.
	Each such separator contains a vertex from $\{v_{1,k}, v_{2,k},v_{3,k}\}$,
	for $k\in K\setminus \{j\}$, since $(a,v_{1,k}, v_{2,k},v_{3,k},v_{4,j},v_{5,j},v_{6,j},b)$ is a path in $G_1$,
	and it contains a vertex from $\{v_{4,k}, v_{5,k},v_{6,k}\}$,
	for $k\in K\setminus \{j\}$, since $(a,v_{1,j}, v_{2,j},v_{3,j},v_{4,k},v_{5,k},v_{6,k},b)$ is a path in $G_1$.
	Due to minimality, the separators in $\mathcal{S}_j$ contain no other vertices.
	Thus, we have that $|\mathcal{S}_j| = 3^{2\cdot(|J|-1)}$.
	We also note that $\mathcal{S}_j \cap \mathcal{S}_k = \emptyset$ if $j\ne k$.
	Therefore, the number of minimal separators of $G_1$ is at least%
	\footnote{There are also minimal $(a,b)$-separators that are completely contained in $V_1\cup V_2\cup V_3$ or $V_4\cup V_5\cup V_6$, but their number does not affect our bound in the first 10 decimal digits in the base of the exponent.}
	$|J|\cdot 3^{2\cdot (|J|-1)} > 2.1271\cdot 10^{23}$.
	
	Minimal $(a,b)$-separators for $G_\ell$ are obtained by taking the union of minimal separators for the different copies of $G_1$.
	Their number is therefore at least $(|J|\cdot 3^{2\cdot (|J|-1)})^\ell = (|J|\cdot 3^{2\cdot (|J|-1)})^{\frac{n-2}{6\cdot |J|}} \in \omega(\lb^n)$, where $n=\ell\cdot 6\cdot |J|+2$ is the number of vertices of $G_\ell$.
	\qed
\end{proof}

\begin{figure}[tb]
	\centering
	\begin{tikzpicture}[xscale=1.1,yscale=0.8]
	\pgfmathtruncatemacro\dimi{6}
	
	\node[vertex,label=above left:$a$] (a) at (0,3) {};
	\node[vertex,label=above right:$b$] (b) at (7,3) {};
	
	\foreach \j in {1,3,4,5} {
		\foreach \i in {1,2,...,\dimi} {
			\node[vertex] (v-\i-\j) at (\i,\j) {};
		}
		\draw (a)--(v-1-\j);
		\draw (v-6-\j)--(b);
		\foreach \i in {1,2,4,5} {
			\pgfmathparse{int(round(\i+1))}
			\draw (v-\i-\j)--(v-\pgfmathresult-\j);
		}
	}
	\foreach \j in {1,3,4,5} {
		\foreach \k in {1,3,4,5} {
			\ifnum\j=\k
			\else
				\draw (v-3-\j)--(v-4-\k);
			\fi
		}
	}
	\node at (2.5,2) {$\vdots$};
	\node at (4.5,2) {$\vdots$};
	
	\end{tikzpicture}
	\caption{\label{fig:lb}The graph $G_1$ has 24 horizontal layers; only 4 are depicted.}
\end{figure}
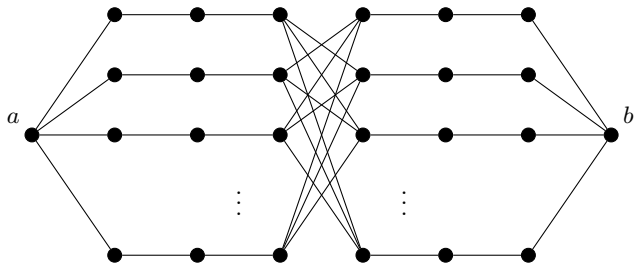

Based on results from~\cite{BouchitteT01}, Bouchitt{\'e} and Todinca~\cite{BouchitteT02} observed
that the number of potential maximal cliques in a graph is at least the number of
minimal separators divided by the number of vertices $n$.
Therefore, we arrive at the following corollary of Theorem~\ref{thm:lb}.

\begin{corollary}\label{cor:pmcs}
	$\pmc(n)\in \omega(\lb^n)$.
\end{corollary}

\section{Conclusion}

We have given a simpler proof for the best known asymptotic upper bound on $\sep(n)$,
and we have improved the best known lower bound from $\Omega(3^{n/3})$ to $\omega(\lb^n)$, thereby reducing the gap between the current best lower and upper bound.
Before our work, it seemed reasonable to believe that $\sep(n)$ could be asymptotically equal to the best known lower bound.
We showed that this is not the case, and we believe there is room to further improve the lower bound.

\section*{Acknowledgments}
NICTA is funded by the Australian Government through the Department of Communications and the Australian Research Council through the ICT Centre of Excellence Program.
Serge Gaspers is the recipient of an Australian Research Council Discovery Early Career Researcher Award (project number DE120101761) and a Future Fellowship (project number FT140100048).

\bibliographystyle{plain}
\bibliography{literature}

\begin{thebibliography}{10}

\bibitem{BerryBC00}
Anne Berry, Jean~Paul Bordat, and Olivier Cogis.
\newblock Generating all the minimal separators of a graph.
\newblock {\em International Journal of Foundations of Computer Science},
  11(3):397--403, 2000.

\bibitem{BouchitteT01}
Vincent Bouchitt\'e and Ioan Todinca.
\newblock Treewidth and minimum fill-in: grouping the minimal separators.
\newblock {\em SIAM Journal on Computing}, 31:212--232, 2001.

\bibitem{BouchitteT02}
Vincent Bouchitt{\'{e}} and Ioan Todinca.
\newblock Listing all potential maximal cliques of a graph.
\newblock {\em Theoretical Computer Science}, 276(1-2):17--32, 2002.

\bibitem{Diestel10}
Reinhard Diestel.
\newblock {\em Graph Theory}.
\newblock Springer, 2010.

\bibitem{Eppstein06}
David Eppstein.
\newblock Quasiconvex analysis of multivariate recurrence equations for
  backtracking algorithms.
\newblock {\em ACM Transactions on Algorithms}, 2(4):492--509, 2006.

\bibitem{FominGPR08}
Fedor~V. Fomin, Serge Gaspers, Artem~V. Pyatkin, and Igor Razgon.
\newblock On the minimum feedback vertex set problem: Exact and enumeration
  algorithms.
\newblock {\em Algorithmica}, 52(2):293--307, 2008.

\bibitem{FominGK09}
Fedor~V. Fomin, Fabrizio Grandoni, and Dieter Kratsch.
\newblock A measure \& conquer approach for the analysis of exact algorithms.
\newblock {\em Journal of the ACM}, 56(5):1--32, 2009.

\bibitem{FominGPS08}
Fedor~V. Fomin, Fabrizio Grandoni, Artem~V. Pyatkin, and Alexey~A. Stepanov.
\newblock Combinatorial bounds via measure and conquer: Bounding minimal
  dominating sets and applications.
\newblock {\em ACM Transactions on Algorithms}, 5(1):1--17, 2008.

\bibitem{FominK10}
Fedor~V. Fomin and Dieter Kratsch.
\newblock {\em Exact Exponential Algorithms}.
\newblock Springer, 2010.

\bibitem{FominKTV08}
Fedor~V. Fomin, Dieter Kratsch, Ioan Todinca, and Yngve Villanger.
\newblock Exact algorithms for treewidth and minimum fill-in.
\newblock {\em {SIAM} Journal on Computing}, 38(3):1058--1079, 2008.

\bibitem{FominV10}
Fedor~V. Fomin and Yngve Villanger.
\newblock Finding induced subgraphs via minimal triangulations.
\newblock In {\em Proceedings of the 27th International Symposium on
  Theoretical Aspects of Computer Science (STACS 2010)}, volume~5 of {\em
  LIPIcs}, pages 383--394. Schloss Dagstuhl - Leibniz-Zentrum fuer Informatik,
  2010.

\bibitem{FominV12}
Fedor~V. Fomin and Yngve Villanger.
\newblock Treewidth computation and extremal combinatorics.
\newblock {\em Combinatorica}, 32(3):289--308, 2012.

\bibitem{Gaspers10}
Serge Gaspers.
\newblock {\em Exponential Time Algorithms: Structures, Measures, and Bounds}.
\newblock VDM Verlag Dr. Mueller e.K., 2010.

\bibitem{GaspersS12}
Serge Gaspers and Gregory~B. Sorkin.
\newblock A universally fastest algorithm for {M}ax 2-{S}at, {M}ax 2-{CSP}, and
  everything in between.
\newblock {\em Journal of Computer and System Sciences}, 78(1):305--335, 2012.

\bibitem{Heggernes06}
Pinar Heggernes.
\newblock Minimal triangulations of graphs: {A} survey.
\newblock {\em Discrete Mathematics}, 306(3):297--317, 2006.

\bibitem{KloksK98}
Ton Kloks and Dieter Kratsch.
\newblock Listing all minimal separators of a graph.
\newblock {\em {SIAM} Journal on Computing}, 27(3):605--613, 1998.

\bibitem{MillerM60}
R.~E. Miller and D.~E. Muller.
\newblock A problem of maximum consistent subsets.
\newblock IBM Research Report RC-240, J. T. Watson Research Center, Yorktown
  Heights, NY, 1960.

\bibitem{MoonM65}
John~W. Moon and Leo Moser.
\newblock On cliques in graphs.
\newblock {\em Israel Journal of Mathematics}, 3:23--28, 1965.

\bibitem{ParraS97}
Andreas Parra and Petra Scheffler.
\newblock Characterizations and algorithmic applications of chordal graph
  embeddings.
\newblock {\em Discrete Applied Mathematics}, 79(1-3):171--188, 1997.

\bibitem{ShenL97}
Hong Shen and Weifa Liang.
\newblock Efficient enumeration of all minimal separators in a graph.
\newblock {\em Theoretical Computer Science}, 180(1-2):169--180, 1997.

\bibitem{Villanger06}
Yngve Villanger.
\newblock Improved exponential-time algorithms for treewidth and minimum
  fill-in.
\newblock In {\em Proceedings of the 7th Latin American Theoretical Informatics
  Symposium (LATIN 2006)}, volume 3887 of {\em Lecture Notes in Computer
  Science}, pages 800--811. Springer-Verlag, Berlin, 2006.

\end{thebibliography}
\end{document}